\documentclass[journal]{IEEEtran}
\usepackage{graphicx}
\usepackage{amsmath}
\usepackage[noend]{algpseudocode}
\usepackage{algorithmicx,algorithm}
\usepackage{amsthm}
\usepackage{amsfonts}
\usepackage{subfigure} 
\usepackage{color}
\usepackage{cite}
\usepackage{setspace}

\newtheorem{Corollary}{Corollary}
\newtheorem{Proposition}{Proposition}
\usepackage{amssymb}
\usepackage{color}
\graphicspath{{images/}}
\usepackage{booktabs}
\usepackage{multirow} 
\usepackage{makecell}
 \usepackage[numbers,sort&compress]{natbib}
 
\usepackage{stfloats}
\usepackage{amsmath}
\usepackage{amssymb}
\usepackage{enumitem}
\usepackage{threeparttable}
\usepackage{comment}
\usepackage{bbm}
\usepackage{bm}
\usepackage{colortbl}
\usepackage{tabularx}
\usepackage[table]{xcolor}
\usepackage{algorithm}
\usepackage{algpseudocode}
\usepackage{url}
\algrenewcommand\algorithmicrequire{\textbf{Input:}}
\algrenewcommand\algorithmicensure{\textbf{Output:}}
\begin{document}	
\title{Kinetic Handover Modeling in Nonstationary Cellular Networks with Fluid Spatiotemporal Stochastic Geometry}
	\author{Wen-Yu Dong,~\IEEEmembership{Member,~IEEE}, Zheng Lin, Jinfeng Xu, Yuxin Zhang, Zihan Fang
		\\ Qi Bi,~\IEEEmembership{Fellow,~IEEE}, and Sheng Chen,~\IEEEmembership{Life Fellow,~IEEE}	%
	
	\thanks{W.-Y. Dong and Q. Bi are with Future Technology Research Center, China Telecom Research Institute, Beijing 102209, China (E-mails: dongwy@chinatelecom.cn; qibi@chinatelecom.cn)} %
	\thanks{Z. Lin is with the Interdisciplinary Centre for Security, Reliability and Trust (SnT), University of Luxembourg, L-1855 Luxembourg, Luxembourg (E-mail: zhenglin@ieee.org).}
	\thanks{J. Xu is with the Department of Electrical and Computer Engineering, The University of British Columbia, Vancouver, Canada (E-mail: jinfeng.xu@ubc.ca).}
	\thanks{ Y. Zhang is with the Institute of Space Internet, Fudan University, Shanghai 200433, China (E-mail: yxzhang24@m.fudan.edu.cn).}
	\thanks{Z. Fang is with the Hong Kong JC STEM Lab of Smart City and Department of Computer Science, City University of Hong Kong, Kowloon, Hong Kong, China (E-mail: zihanfang3-c@my.cityu.edu.hk).}
	\thanks{S. Chen is with School of Electronics and Computer Science, University of Southampton, Southampton SO17 1BJ, U.K. (E-mail: sqc@ecs.soton.ac.uk).} %
	\vspace*{-5mm}
}

%	\markboth{Journal of \LaTeX\ Class Files,~Vol.~14, No.~8, August~2021}%
%	{Shell \MakeLowercase{\textit{et al.}}: Stochastic Geometry Based Modeling and Analysis of Cooperative Satellite-Aerial-Terrestrial Systems for Nomadic Communications with Weak Signal Coverage}
	
	%	\IEEEpubid{0000--0000/00\$00.00~\copyright~2021 IEEE}
	% column for its text to clear the IEEEp  ubid mark.
	
\maketitle 
\vspace{-1cm}

\begin{abstract}
	Inferring directional handover traffic from aggregate mobility observations remains challenging because population density evolution does not fully describe the movement of individual users across cell boundaries. This paper develops a kinetic framework for geometry-induced serving-region crossings in nonstationary cellular networks, connecting boundary-normal velocity statistics with fluid spatiotemporal stochastic geometry (F-STSG) for density-based transport reconstruction. The framework distinguishes signed net population transfer from gross bidirectional crossings and reveals two information limits: density evolution does not uniquely determine physical pairwise transport, while first-order transport alone cannot resolve counter-moving crossings. We establish exact recovery conditions for potential transport and develop coherent and Gaussian mobility closures for estimating directional crossing rates. Sharp moment-constrained bounds further characterize the remaining uncertainty under bounded normal velocities without assuming a specific velocity distribution. An error bound links directional estimation accuracy to transport reconstruction and unresolved mobility information. The formulation also recovers the classical stationary Poisson--Voronoi handover law. Numerical experiments examine the information limits, reconstruction sensitivity, and closure accuracy. On 47 held-out vehicular simulation windows, F-STSG achieves lower net-flow reconstruction errors than three baselines, while auxiliary velocity statistics improve directional estimates under matched Gaussian mobility.
\end{abstract}


\section{Introduction}\label{sec:introduction} % S1

Mobility management is a core function of cellular networks because changes in serving-cell association affect signaling overhead, service continuity, and radio resource utilization~\cite{Shayea2020Survey,Gures2020Survey}. Network densification, heterogeneous cell sizes, and high-frequency operation further increase the frequency and complexity of mobility-related association changes \cite{Tabassum2019Mobility,lin2021spatial}. Handover (HO) traffic can consequently vary substantially across neighboring cells and over time, particularly in the presence of collective mobility associated with commuting corridors, transportation hubs, and large public events. Characterizing when, where, and in which direction such association changes arise is therefore relevant to mobility-aware network dimensioning and control. At the trajectory-prediction level, diffusion-based methods have also been developed to forecast individual agent motion~\cite{yang2025trajdiff,zhang2026large}. Such methods address user-level motion prediction, whereas the present work characterizes directional cell-boundary crossings from aggregate population observations.

The modeling difficulty arises from the information lost between individual boundary crossings and aggregate mobility observations. At the microscopic level, a geometry-induced serving-region transition occurs when an individual user equipment (UE) crosses a boundary between neighboring serving regions. At the macroscopic level, however, a population density field retains neither individual trajectories nor the complete local velocity distribution. This leads to two distinct information losses. First, the evolution of the population density constrains the divergence of the physical population flux but does not, in general, uniquely determine the physical transfer between individual cell pairs. Second, even if the physical first-order flux is known, counter-moving users can cancel in the signed flux while still contributing to the total number of boundary crossings. A population-level HO model must therefore distinguish signed net transfer from gross crossing traffic and, separately, distinguish the physical mobility flux from a transport field reconstructed from aggregate observations.

This paper develops a kinetic characterization of mobility-induced serving-region crossings. For each neighboring cell pair, the directional crossing rate is defined from the local population density and the positive part of the boundary-normal UE velocity. This representation separates the signed net transfer between two cells from their gross bidirectional crossing traffic. The former is determined by the boundary-normal component of the physical first-order population flux, whereas the latter generally depends on additional local velocity information. The resulting formulation makes explicit which HO quantities are determined by first-order population transport and which remain unresolved when only aggregate motion is available.

Fluid spatiotemporal stochastic geometry (F-STSG) \cite{Dong2026FSTSG} is then used to connect this characterization with density-based observation. From successive spatial density snapshots, F-STSG constructs a canonical first-order transport field that is consistent with the observed density evolution through a minimum-energy formulation. This reconstructed field is not assumed to coincide with the physical flux in general. Rather, the paper identifies conditions under which exact recovery is possible and quantifies how transport-reconstruction error propagates to directional crossing estimates. When counter-flow effects are relevant, additional local velocity statistics are introduced through a dispersion-aware closure. The baseline model concerns geometry-induced serving-region crossings. Radio measurements, hysteresis, time-to-trigger, target selection, and other protocol mechanisms governing executed HOs are outside the crossing model itself.

\vspace*{-1mm}
\subsection{Related Work}\label{subsec:related_work} % S1.1

Analytical studies of mobility in cellular networks commonly use stochastic geometry to relate user motion, association geometry, and HO statistics. For hexagonal and Poisson--Voronoi deployments, classical results established the dependence of the mean HO rate on UE speed and base-station (BS) density and derived corresponding cell residence and sojourn-time statistics \cite{Lin2013Mobility}. Subsequent extensions considered biased multi-tier association, general UE trajectories, network densification, and control-plane/user-plane separation \cite{Sadr2015Handoff,Bao2015Mobility,Ibrahim2016Dense,Salehi2021Sojourn}. Related analyses addressed movement-aware cooperation and mobility in ultra-dense and radio-frequency/visible-light-communication (RF/VLC) networks \cite{Sun2021MACH,Arshad2021RFVLC,Arshad2022Sojourn}, with a broader overview of mobility-aware stochastic geometry given in \cite{Tabassum2019Mobility}. The methodology has also been applied beyond terrestrial cellular networks, including satellite--aerial, cooperative satellite--unmanned-aerial-vehicle (UAV), and heterogeneous non-terrestrial systems \cite{10058140,9586571,DongJSAC,DongTCOM,teng2026communication,zhang2026resource,lin2024fedsn}. Despite the different network geometries considered, mobility in these models is generally specified through individual UE trajectories or prescribed mobility distributions rather than through a population field evolving jointly in space and time.

Beyond association geometry, HO models have incorporated radio propagation and mobility-management mechanisms. Multi-connectivity, blockage, beam management, and directional transmission have been considered in millimeter-wave, terahertz, and vehicular networks \cite{Ozkoc2021Multiconnectivity,Kalamkar2022Beam,Aghashahi2023Beam,lin2022tracking}, while more detailed association models have been used to characterize HO trigger regions, failures, and ping-pong behavior \cite{Wei2025IRSHO}. 
At the protocol and control level, previous studies have optimized hysteresis and time-to-trigger parameters, incorporated signal prediction and target-cell selection, and applied reinforcement learning and open radio access network (O-RAN) methods to mobility management~\cite{Huang2022Handover,Karmakar2023Mobility,Prado2023Mobility,Kalntis2026Meta,duan2026llm}.
 These models describe how radio measurements, association rules, and protocol parameters affect HO decisions and execution. Their primary state variables, however, remain user-level measurements, trajectories, or serving-cell decisions rather than the aggregate transport of a mobile population.

\begin{table*}[t]
\vspace*{-1mm}
	\centering
	\caption{Comparison of representative approaches with respect to the mobility and HO modeling aspects considered in this paper.}
	\label{tab:related_ho} % Tab.I
	\vspace*{-2mm}
	\footnotesize
	\renewcommand{\arraystretch}{1.20}
\setlength{\tabcolsep}{6.0pt}
%\begin{tabular*}{\textwidth}{@{\extracolsep{\fill}}p{6.7cm}cccc@{}}
\begin{tabular}{lcccc}
		\toprule
		Representative approaches &
		\shortstack{Directed cell-pair\\mobility} &
		\shortstack{Time-varying\\population field} &
		\shortstack{Net/gross\\distinction} &
		\shortstack{Density-to-HO\\reconstruction}\\
		\midrule
		
		Stochastic-geometry mobility and HO modeling
		\cite{Lin2013Mobility,Sadr2015Handoff,Bao2015Mobility,Ibrahim2016Dense,Salehi2021Sojourn,Sun2021MACH,Arshad2021RFVLC,Arshad2022Sojourn}
		&
		$\triangle$ &
		-- &
		-- &
		--\\
		
		Radio- and beam-aware HO modeling
		\cite{Ozkoc2021Multiconnectivity,Kalamkar2022Beam,Aghashahi2023Beam,Wei2025IRSHO}
		&
		$\triangle$ &
		-- &
		-- &
		--\\
		
		HO control, prediction, and learning
		\cite{Huang2022Handover,Karmakar2023Mobility,Prado2023Mobility,Kalntis2026Meta}
		&
		$\checkmark$ &
		-- &
		-- &
		--\\
		
		Empirical HO measurement and mobility analysis
		\cite{Demissie2013Handover,Kalntis2024Telco}
		&
		$\checkmark$ &
		$\triangle$ &
		-- &
		--\\
		
		F-STSG framework
		\cite{Dong2026FSTSG}
		&
		-- &
		$\checkmark$ &
		-- &
		--\\
		
		Field-based F-STSG applications
		\cite{Dong2026TMC,Dong2026VDR,DongFSAT}
		&
		-- &
		$\triangle$ &
		-- &
		--\\
		
		\midrule
		This work
		&
		$\checkmark$ &
		$\checkmark$ &
		$\checkmark$ &
		$\checkmark$\\
		
		\bottomrule
\end{tabular}
	
	\vspace{1mm}
	\textit{Note:} HO: handover; F-STSG: fluid spatiotemporal stochastic geometry. $\checkmark$: explicitly treated; $\triangle$: partially or indirectly treated; --: not explicitly addressed.
	\vspace*{-4mm}
\end{table*}

Large-scale measurement studies approach the problem from observed HO records rather than assumed mobility models. Cellular HO traces have been used to infer urban mobility patterns and relate inter-cell transitions to population activity \cite{Demissie2013Handover}. More recent countrywide measurements have characterized the spatial and temporal distributions of horizontal and vertical HOs and examined their relation to population distribution and network conditions \cite{Kalntis2024Telco,wang2026exploring}. Such measurements provide direct evidence of how association changes appear in operational networks and complement analytical models based on idealized mobility assumptions. In these studies, however, the HO event is already part of the observation. They do not determine how directional cell-pair crossing rates can be obtained from an evolving population field before individual HO events are observed, nor which components of the resulting crossing traffic are identifiable from aggregate population information.

F-STSG \cite{Dong2026FSTSG} starts from a different state description, representing a nonstationary network by a time-varying spatial density and a macroscopic transport field linked through a continuity equation. This field representation has subsequently been applied to flux-aware UAV infrastructure provisioning \cite{Dong2026TMC}, vorticity-based routing in ultra-dense networks \cite{Dong2026VDR}, and fluid-dynamic interference modeling for low-Earth-orbit (LEO) mega-constellations \cite{DongFSAT}. These studies describe network dynamics at the field level but do not formulate directional mobility-induced crossings between neighboring serving regions. Density evolution alone does not specify the physical population transfer between individual cell pairs, while first-order population transport does not, in general, determine the total crossing traffic generated by counter-moving users. The present work addresses these two distinctions by connecting microscopic boundary crossings with macroscopic population transport and by establishing when directional HO traffic can be reconstructed from successive density observations.

\begin{figure*}[b]
\vspace*{-5mm}
	\centering
	\includegraphics[width=0.75\textwidth]{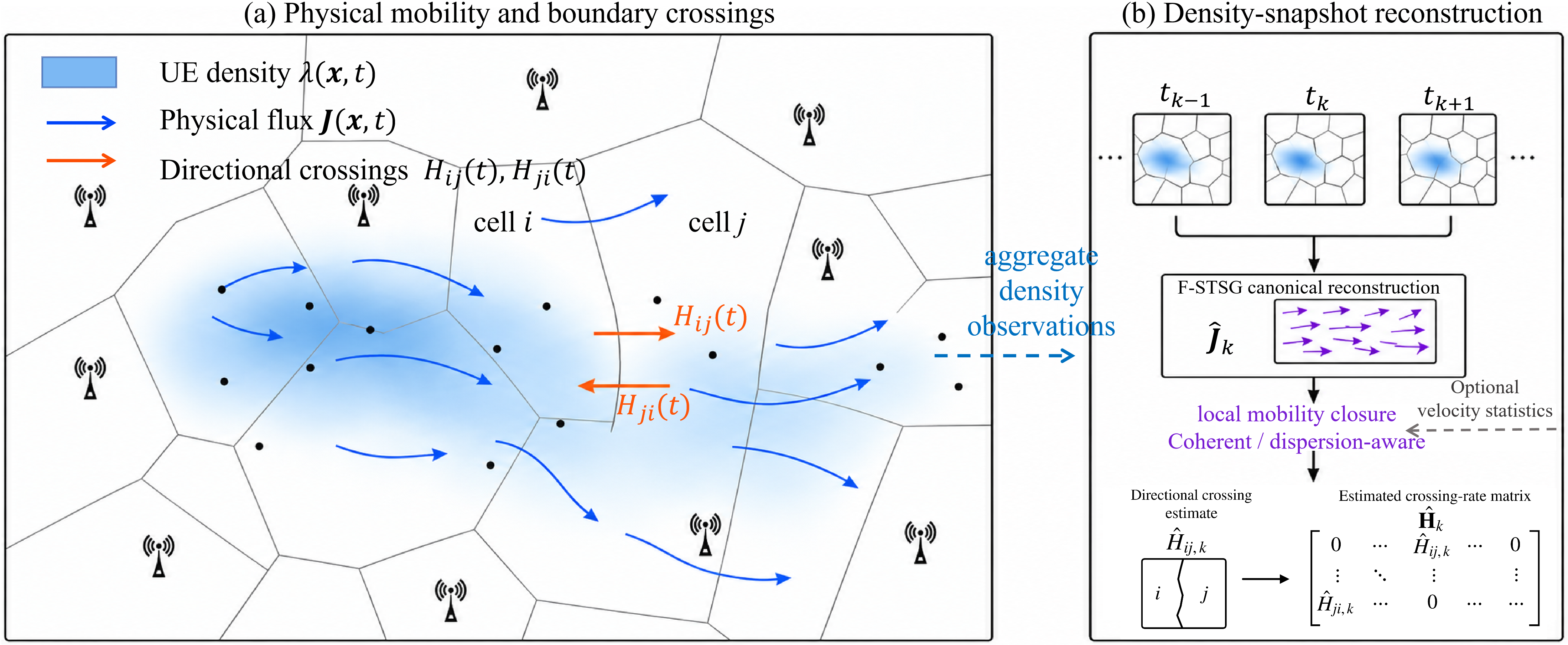}
	\vspace*{-3mm}
	\caption{System model for mobility-induced serving-region crossings and density-snapshot-based reconstruction. Physical population transport generates directional crossings between neighboring serving regions, whereas aggregate density observations are used to reconstruct first-order transport and estimate the corresponding crossing rates under a local mobility closure.}
	\label{fig:system_model} % Fig.1
	\vspace*{-1mm}
\end{figure*}

\vspace*{-2mm}
\subsection{Contributions and Organization}\label{subsec:contributions} % S1.2

The main contributions are summarized as follows.
\begin{itemize}
	\item We develop a kinetic formulation of directional serving-region crossings based on boundary-normal UE velocity moments. It distinguishes signed net population transfer from gross bidirectional crossings, preserves cell-level population balance, and recovers the classical stationary Poisson--Voronoi crossing-rate law as a special case.
	
	\item We identify two information limits: density evolution does not uniquely determine physical pairwise transport, and first-order transport does not generally determine gross crossings. We establish exact recovery conditions for potential transport and derive sharp distribution-free bounds on directional crossing rates given the first two normal-velocity moments and a bounded-speed constraint. Coherent and Gaussian mobility closures provide conditional estimates under additional mobility information.
	
	\item We connect the kinetic formulation with density-snapshot-based F-STSG reconstruction and establish an error bound accounting for reconstruction error, physical-flux ambiguity, auxiliary-statistic error, and closure mismatch. Controlled experiments examine these effects, while held-out vehicular simulations compare net-flow reconstruction against three baselines and evaluate the role of auxiliary velocity statistics.
\end{itemize}

Table~\ref{tab:related_ho} compares the proposed approach with the closely related studies in terms of mobility and HO modeling aspects.

The remainder of this paper is organized as follows. Section~\ref{sec:system_model} introduces the system model and directional crossing definitions. Section~\ref{sec:ho_characterization} establishes the kinetic relations, information limits, recovery conditions, and mobility closures. Section~\ref{sec:estimation} develops density-based reconstruction and directional crossing estimation. Section~\ref{sec:numerical_results} presents numerical results, and Section~\ref{sec:conclusion} concludes the paper.

\begin{table}[h]
\vspace*{-3mm}
	\centering
	\caption{Key notation.}
	\label{tab:key_notation} % Tab.II
	\vspace*{-2mm}
	\scriptsize
	\renewcommand{\arraystretch}{1.15}
	\setlength{\tabcolsep}{4.0pt}
	\begin{tabular}{ll}
		\toprule
		Symbol & Definition\\
		\midrule
		$\Omega$, $\mathcal{V}_i$
		& Service region and serving region of BS $i$\\
		$\Gamma_{ij}$, $\mathbf{n}_{ij}$
		& Common cell boundary and unit normal from $i$ to $j$\\
		$\lambda(\mathbf{x},t)$
		& Spatial UE population density\\
		$N_i(t)$
		& UE population in cell $i$\\
		$\mathbf{J}(\mathbf{x},t)$
		& Physical first-order population flux\\
		$\widehat{\mathbf{J}}(\mathbf{x},t)$
		& Canonical flux reconstructed from exact density evolution\\
		$U_{ij}$
		& UE velocity component normal to $\Gamma_{ij}$\\
		$H_{ij}(t)$
		& Directional crossing rate from cell $i$ to $j$\\
		$F_{ij}(t)$
		& Signed net flow, $H_{ij}-H_{ji}$\\
		$G_{ij}(t)$
		& Gross crossing rate, $H_{ij}+H_{ji}$\\
		$m_{ij}$, $s_{ij}^{2}$
		& Conditional mean and variance of normal velocity\\
		$G_{ij}^{\mathrm L}$, $G_{ij}^{\mathrm U}$
		& Lower and upper bounds on gross crossing rate\\
		$\lambda_k$, $q_k$
		& Density snapshot and corrected temporal derivative\\
		$\widehat{\mathbf{J}}_k$
		& Flux reconstructed from density snapshots\\
		$\widehat F_{ij,k}$, $\widehat H_{ij,k}$
		& Estimated signed net flow and directional crossing rate\\
		\bottomrule
	\end{tabular}
	\vspace*{-4mm}
\end{table}

\subsection{Notation}\label{S1.3}

Scalars are denoted by italic letters, vectors by bold lowercase letters, and sets or spatial regions by calligraphic letters. The operators $[\cdot]_+$, $\mathbb{E}[\cdot]$, $\nabla$, and $\nabla\!\cdot$ denote the positive part, expectation, gradient, and divergence, respectively. A hat denotes a reconstructed or estimated quantity, and the corresponding unhatted symbol denotes the reference quantity unless stated otherwise. The principal symbols used throughout the paper are summarized in Table~\ref{tab:key_notation}, while symbols used only in individual derivations are defined locally when introduced.

%\vspace*{-1mm}
%========================II=================================
\section{System Model and Problem Formulation}\label{sec:system_model} % S2

We consider mobility-induced association changes in a cellular network from two complementary descriptions. The physical population flux characterizes the first-order transport generated by UE motion, whereas F-STSG reconstructs a canonical transport field from aggregate density evolution. Directional crossings of neighboring serving regions are then defined from the local boundary-normal velocity distribution. Fig.~\ref{fig:system_model} summarizes the resulting model: physical UE motion generates serving-region crossings, while successive density observations are used to reconstruct first-order transport and, under an appropriate local mobility closure, estimate the corresponding crossing rates.

\vspace*{-1mm}
\subsection{Cellular Geometry and Population Transport}\label{subsec:cellular_mobility} % S2.1

Consider a bounded planar service region $\Omega\subset\mathbb{R}^{2}$ containing $M$ fixed BSs at distinct locations $\{\mathbf{b}_{i}\}_{i\in\mathcal{I}}$, where $\mathcal{I}=\{1,\ldots,M\}$. To isolate mobility-induced association changes from radio fluctuations and protocol filtering, the baseline model adopts nearest-BS association. The serving region of BS $i$ is
\begin{equation} % eq.1
	\mathcal{V}_{i}
	=
	\left\{
	\mathbf{x}\in\Omega:
	\|\mathbf{x}-\mathbf{b}_{i}\|_{2}
	\le
	\|\mathbf{x}-\mathbf{b}_{j}\|_{2},
	\ \forall j\in\mathcal{I}
	\right\}.
	\label{eq:voronoi_cell}
\end{equation}
The regions $\{\mathcal{V}_{i}\}_{i\in\mathcal{I}}$ form a partition of $\Omega$ up to their boundaries. Two cells are defined as neighbors when their common boundary has nonzero length. For neighboring cells $i$ and $j$, let
$\Gamma_{ij}=\partial\mathcal{V}_{i}\cap\partial\mathcal{V}_{j}\cap\Omega$
denote the common boundary and let $\mathbf{n}_{ij}$ denote the unit normal directed from $\mathcal{V}_{i}$ to $\mathcal{V}_{j}$. Hence,
$\Gamma_{ij}=\Gamma_{ji}$ and $\mathbf{n}_{ji}=-\mathbf{n}_{ij}$. The neighbor set of cell $i$ is denoted by $\mathcal{N}_{i}$.

Let $\lambda(\mathbf{x},t)\geq0$ denote the spatial UE population density at location $\mathbf{x}$ and time $t$. The UE population in cell $i$, $N_{i}(t)$, the local conditional mean velocity, $\mathbf{v}(\mathbf{x},t)$, and the corresponding physical population flux, $\mathbf{J}(\mathbf{x},t)$, are given respectively:
\begin{align} % eqs.2,3
	N_{i}(t)
	&=
	\int_{\mathcal{V}_{i}}
	\lambda(\mathbf{x},t)
	\,\mathrm{d}\mathbf{x},
	\label{eq:cell_population}\\
	\mathbf{v}(\mathbf{x},t)
	&=
	\mathbb{E}
	\left[
	\mathbf{U}
	\mid
	\mathbf{x},t
	\right],
	\qquad
	\mathbf{J}(\mathbf{x},t)
	=
	\lambda(\mathbf{x},t)\mathbf{v}(\mathbf{x},t),
	\label{eq:population_flux}
\end{align}
where $\mathbf{U}$ denotes the velocity of a randomly selected UE conditioned on location and time. The conditional velocity distribution is assumed to have a finite first moment. Thus, $\mathbf{J}$ retains the first-order population transport induced by the conditional mean velocity but not the complete distribution of individual UE velocities.

The baseline model considers population redistribution within $\Omega$ without physical arrivals or departures through its outer boundary. Assuming sufficient regularity for the differential and boundary operations, the density and physical flux satisfy
\begin{equation} % eq.4
	\begin{aligned}
		\partial_t\lambda(\mathbf{x},t)
		+\nabla\!\cdot\mathbf{J}(\mathbf{x},t)
		&=0,
		&&\mathbf{x}\in\Omega,\\
		\mathbf{J}(\mathbf{x},t)\cdot\mathbf{n}_{\Omega}(\mathbf{x})
		&=0,
		&&\mathbf{x}\in\partial\Omega.
	\end{aligned}
	\label{eq:continuity_physical}
\end{equation}
where $\mathbf{n}_{\Omega}$ denotes the outward unit normal on $\partial\Omega$. The no-flux boundary condition implies conservation of the total UE population in $\Omega$. Open boundaries or physical arrivals and departures can instead be represented through boundary fluxes or source terms but are not included in the baseline model.

Equation~\eqref{eq:continuity_physical} determines the divergence of the physical flux but not the complete field $\mathbf{J}$. F-STSG~\cite{Dong2026FSTSGKinematics} selects a canonical flux by minimizing the weighted kinetic energy $\int_{\Omega}\|\mathbf{j}\|^2/\lambda\,\mathrm{d}\mathbf{x}$ over admissible flux fields $\mathbf{j}$ satisfying the continuity equation and no-flux boundary condition. Under the regularity conditions of the weighted elliptic problem, the potential $\phi(\mathbf{x},t)$ and reconstructed flux $\widehat{\mathbf{J}}(\mathbf{x},t)$ satisfy
\begin{equation} % eq.5
	\begin{aligned}
		\nabla\!\cdot
		\left(
		\lambda(\mathbf{x},t)
		\nabla\phi(\mathbf{x},t)
		\right)
		&=
		\frac{\partial\lambda(\mathbf{x},t)}{\partial t},
		\qquad
		\mathbf{x}\in\Omega,\\
		\lambda(\mathbf{x},t)
		\nabla\phi(\mathbf{x},t)
		\cdot
		\mathbf{n}_{\Omega}(\mathbf{x})
		&=
		0,
		\qquad
		\mathbf{x}\in\partial\Omega,\\
		\widehat{\mathbf{J}}(\mathbf{x},t)
		&=
		-\lambda(\mathbf{x},t)
		\nabla\phi(\mathbf{x},t).
	\end{aligned}
	\label{eq:fstsg_reconstruction}
\end{equation}
For the closed system, the Neumann compatibility condition follows from conservation of the total population. The potential is determined up to an additive constant, which does not affect $\widehat{\mathbf{J}}$. Throughout the paper, $\mathbf{J}$ denotes the physical first-order population flux, whereas $\widehat{\mathbf{J}}$ denotes the F-STSG reconstruction inferred from density evolution. The two quantities need not coincide in general, even though both are consistent with the same population balance. In particular, density evolution alone does not determine the physical transport across each individual cell boundary.

\vspace*{-1mm}
\subsection{Boundary-Crossing Model}\label{subsec:boundary_model} % S3.1

Under nearest-BS association, a continuous UE trajectory changes serving region when it crosses the common boundary of two neighboring cells. We refer to this event as a geometry-induced serving-region crossing and denote its directional rate from cell $i$ to cell $j$ by $H_{ij}(t)$. This rate characterizes mobility-induced association crossings rather than protocol-level HO execution, which may additionally depend on radio measurements and mobility-management parameters.

For neighboring cells $i$ and $j$, define the boundary-normal component of the local UE velocity as
\begin{equation} % eq.6
	U_{ij}
	=
	\mathbf{U}^{\mathsf T}\mathbf{n}_{ij}.
	\label{eq:normal_velocity}
\end{equation}
A positive value of $U_{ij}$ represents motion from $\mathcal{V}_{i}$ toward $\mathcal{V}_{j}$, whereas a negative value represents motion in the reverse direction. Assuming continuous trajectories and an integrable first normal-velocity moment along $\Gamma_{ij}$, the directional crossing rate is
\begin{equation} % eq.7
	H_{ij}(t)
	\!=\!\!
	\int_{\Gamma_{ij}}  \!\!\!\!\!
	\lambda(\mathbf{x},t)
	\mathbb{E}
	\left[ 
	[U_{ij}]_{+} \!
	\mid \!
	\mathbf{x},t
	\right]
	\,\mathrm{d}\ell,
	\,
	[a]_{+}\! =\! \max\{a,0\},
	\label{eq:directional_ho_rate}
\end{equation}
where $\mathrm{d}\ell$ denotes the line element along $\Gamma_{ij}$. Reversing the boundary normal gives the corresponding rate $H_{ji}$. Equation~\eqref{eq:directional_ho_rate} has units of crossings per unit time and depends only on the velocity component normal to the serving-region boundary. Tangential motion does not contribute to the crossing rate. Intersections shared by three or more Voronoi cells have zero line measure and do not contribute to the boundary integral.

The problem considered in this paper is to determine which directional crossing quantities can be characterized or estimated from aggregate population observations without reconstructing individual UE trajectories. The physical first-order flux $\mathbf{J}$ contains only the conditional mean velocity in \eqref{eq:population_flux}, whereas \eqref{eq:directional_ho_rate} depends on a one-sided moment of the boundary-normal velocity distribution. Moreover, density evolution does not, in general, uniquely determine the physical pairwise transport represented by $\mathbf{J}$. The next section therefore separates the information contained in density evolution, first-order population transport, and unresolved local mobility, and identifies the conditions under which directional crossing rates can be recovered.

\vspace*{-1mm}
%========================III=================================
\section{Kinetic Handover Characterization}\label{sec:ho_characterization} % S3

The directional crossing rate in \eqref{eq:directional_ho_rate} depends on both first-order population transport and the local distribution of boundary-normal UE velocities. This section first relates directional crossings to the physical population flux and then identifies two distinct information limitations. The first concerns the recovery of physical pairwise transport from density evolution, whereas the second concerns the recovery of gross crossing traffic from first-order transport. Conditions under which these ambiguities can be resolved are then established, followed by a stationary isotropic limit that recovers the classical Poisson--Voronoi HO result.

\vspace*{-1mm}
\subsection{Net and Gross Crossing Flows}\label{subsec:net_gross} % S3.1

The two directional crossing rates associated with a common cell boundary contain different information about population transport. Define the total outgoing and incoming crossing rates of cell $i$ respectively as
\begin{equation} % eq.8
	H_i^{\mathrm{out}}(t)
	\triangleq
	\sum_{j\in\mathcal{N}_i}H_{ij}(t),
	\qquad
	H_i^{\mathrm{in}}(t)
	\triangleq
	\sum_{j\in\mathcal{N}_i}H_{ji}(t).
	\label{eq:cell_in_out_rates}
\end{equation}

\begin{Proposition} 
	\label{prop:boundary_flow_decomposition}
	For any neighboring cells $i$ and $j$ satisfying the assumptions of \eqref{eq:directional_ho_rate}, define the signed net flow $F_{ij}(t)$ and gross crossing rate $G_{ij}(t)$ by 
	\begin{equation} % eq.9
		\begin{aligned}
			F_{ij}(t)
			\!&\triangleq\!
			H_{ij}(t)\!-\! H_{ji}(t)
			\!=\!
			\int_{\Gamma_{ij}}
			\mathbf{J}(\mathbf{x},t)
			\cdot
			\mathbf{n}_{ij}
			\,\mathrm{d}\ell,\\
			G_{ij}(t)
			\!&\triangleq\!
			H_{ij}(t)\!+\!H_{ji}(t)
			\!=\!
			\int_{\Gamma_{ij}}
			\lambda(\mathbf{x},t)
			\mathbb{E}
			\left[
			|U_{ij}|
			\mid
			\mathbf{x},t
			\right]
			\,\mathrm{d}\ell.
		\end{aligned}
		\label{eq:net_gross_flows}
	\end{equation}
	Then
	\begin{equation} % eq.10
		G_{ij}(t)
		\geq
		\int_{\Gamma_{ij}}
		\left|
		\mathbf{J}(\mathbf{x},t)
		\cdot
		\mathbf{n}_{ij}
		\right|
		\,\mathrm{d}\ell
		\geq
		|F_{ij}(t)|,
		\label{eq:gross_lower_bound}
	\end{equation}
	and the cell population satisfies
	\begin{equation} % eq.11
		\frac{\mathrm{d}N_i(t)}{\mathrm{d}t}
		=
		H_i^{\mathrm{in}}(t)
		-
		H_i^{\mathrm{out}}(t)
		=
		-\sum_{j\in\mathcal{N}_i}
		F_{ij}(t).
		\label{eq:cell_conservation}
	\end{equation}
\end{Proposition}

\begin{proof}
	Because $\mathbf{n}_{ji}=-\mathbf{n}_{ij}$, we have $U_{ji}=-U_{ij}$. The identities $[a]_+-[-a]_+=a$ and $[a]_++[-a]_+=|a|$, together with $\mathbf{J}=\lambda\mathbb{E}[\mathbf{U}\mid\mathbf{x},t]$, give \eqref{eq:net_gross_flows}. The pointwise inequality
	$\left|\mathbb{E}[U_{ij}\mid\mathbf{x},t]\right|
	\leq
	\mathbb{E}[|U_{ij}|\mid\mathbf{x},t]$
	gives the first inequality in \eqref{eq:gross_lower_bound}, and the triangle inequality gives the second. Integrating \eqref{eq:continuity_physical} over $\mathcal{V}_i$ and applying the divergence theorem, with the outer no-flux condition where $\partial\mathcal{V}_i$ intersects $\partial\Omega$, yields \eqref{eq:cell_conservation}.
\end{proof}

The quantity $F_{ij}$ is the signed net population transfer from cell $i$ to cell $j$, whereas $G_{ij}$ is the gross crossing rate across their common boundary. By construction,
$F_{ji}=-F_{ij}$ and $G_{ji}=G_{ij}$, and
\begin{equation} % eq.12
	H_{ij}
	=
	\frac{G_{ij}+F_{ij}}{2},
	\qquad
	H_{ji}
	=
	\frac{G_{ij}-F_{ij}}{2}.
	\label{eq:directional_from_net_gross}
\end{equation}
Equality in the first inequality of \eqref{eq:gross_lower_bound} holds when the conditional normal velocity has one sign almost surely at almost every point of $\Gamma_{ij}$. Equality in the second additionally requires the mean normal flux to have a common sign along the boundary.

When only the density evolution is available, the physical flux $\mathbf{J}$ is not directly observed. Let the reconstructed pairwise flow associated with \eqref{eq:fstsg_reconstruction} be
\begin{equation} % eq.13
	\widehat F_{ij}(t)
	\triangleq
	\int_{\Gamma_{ij}}
	\widehat{\mathbf{J}}(\mathbf{x},t)
	\cdot
	\mathbf{n}_{ij}
	\,\mathrm{d}\ell.
	\label{eq:reconstructed_pairwise_flow}
\end{equation}

\begin{Proposition} 
	\label{prop:pairwise_transport_ambiguity}
	For a given density evolution and the same outer no-flux condition,
	\begin{equation} % eq.14
		\sum_{j\in\mathcal{N}_i}
		\widehat F_{ij}(t)
		=
		\sum_{j\in\mathcal{N}_i}
		F_{ij}(t)
		=
		-\frac{\mathrm{d}N_i(t)}{\mathrm{d}t}.
		\label{eq:reconstructed_flux_balance}
	\end{equation}
	The individual physical pairwise flows are not uniquely determined by the density evolution in general. In particular, let $\mathbf{C}(\mathbf{x},t)$ be a sufficiently regular field such that $\mathbf{J}+\mathbf{C}$ remains an admissible population flux and
	\begin{equation} % eq.15
		\nabla\!\cdot\mathbf{C}(\mathbf{x},t)=0,
		\qquad
		\mathbf{C}(\mathbf{x},t)
		\cdot
		\mathbf{n}_{\Omega}(\mathbf{x})=0,
		\quad
		\mathbf{x}\in\partial\Omega.
		\label{eq:divergence_free_component}
	\end{equation}
	Then $\mathbf{J}+\mathbf{C}$ produces the same density evolution and outer boundary flux as $\mathbf{J}$, while its flow across $\Gamma_{ij}$ changes by
	\begin{equation} % eq.16
		\int_{\Gamma_{ij}}
		\mathbf{C}(\mathbf{x},t)
		\cdot
		\mathbf{n}_{ij}
		\,\mathrm{d}\ell.
		\label{eq:pairwise_ambiguity}
	\end{equation}
\end{Proposition}

\begin{proof}
	Both $\mathbf{J}$ and $\widehat{\mathbf{J}}$ satisfy the same continuity equation and outer no-flux condition. Integrating their divergences over $\mathcal{V}_i$ and applying the divergence theorem gives \eqref{eq:reconstructed_flux_balance}. A field satisfying \eqref{eq:divergence_free_component} changes neither the divergence of the transport field nor its normal component on the outer boundary. Hence, provided that $\mathbf{J}+\mathbf{C}$ remains admissible, it produces the same density evolution. Its contribution to the flow across an individual common boundary is exactly \eqref{eq:pairwise_ambiguity}.
\end{proof}

Proposition~\ref{prop:pairwise_transport_ambiguity} separates cell population balance from physical pairwise transport. Density evolution fixes the former, but divergence-free transport components can alter the latter without changing the observed density evolution.

The pairwise transport ambiguity is resolved for the following potential-transport class.

\begin{Corollary}
	\label{cor:potential_exact_recovery}
	Assume that the physical population flux admits the representation
	\begin{equation} % eq.17
		\mathbf{J}(\mathbf{x},t)
		=
		-\lambda(\mathbf{x},t)
		\nabla\psi(\mathbf{x},t),
		\label{eq:physical_potential_transport}
	\end{equation}
	satisfies the same outer no-flux condition as \eqref{eq:fstsg_reconstruction}, and assume that the corresponding weighted Neumann problem admits a solution that is unique up to an additive constant. Then
	\begin{equation} % eq.18
		\widehat{\mathbf{J}}(\mathbf{x},t)
		=
		\mathbf{J}(\mathbf{x},t).
		\label{eq:potential_exact_recovery}
	\end{equation}
	Consequently, $\widehat F_{ij}(t)=F_{ij}(t)$ for every neighboring cell pair.
\end{Corollary}

\begin{proof}
	Substituting \eqref{eq:physical_potential_transport} into \eqref{eq:continuity_physical} gives
	\[
	\nabla\!\cdot
	\left[
	\lambda(\mathbf{x},t)
	\nabla\psi(\mathbf{x},t)
	\right]
	=
	\frac{\partial\lambda(\mathbf{x},t)}{\partial t}.
	\]
	The physical no-flux condition also gives
	$\lambda\nabla\psi\cdot\mathbf{n}_{\Omega}=0$
	on $\partial\Omega$. Hence, $\psi$ and $\phi$ satisfy the same weighted Neumann problem. Uniqueness up to an additive constant gives $\nabla\psi=\nabla\phi$, and therefore $\mathbf{J}=\widehat{\mathbf{J}}$.
\end{proof}

The preceding result identifies a class for which the F-STSG reconstruction coincides with the physical first-order transport. Outside this class, agreement of cell population balance does not imply agreement of individual pairwise flows.

A second information limitation remains even when the physical first-order population flux is known.

\begin{Proposition} 
	\label{prop:gross_crossing_nonidentifiability}
	Suppose that $\Gamma_{ij}$ contains a subset of positive line measure on which $\lambda(\mathbf{x},t)>0$. Then the local density $\lambda$ and physical first-order population flux $\mathbf{J}$ do not, in general, uniquely determine the gross crossing rate $G_{ij}$.
\end{Proposition}

\begin{proof}
	Consider a measurable subset of $\Gamma_{ij}$ with positive line measure and positive density. On this subset, let two conditional normal velocity distributions be
	$U_1\sim\operatorname{Unif}[-a,a]$
	and
	$U_2\sim\operatorname{Unif}[-b,b]$,
	where $a,b>0$ and $a\neq b$, while keeping the tangential conditional mean and all mobility statistics outside the subset identical. Both normal distributions have zero mean, and they produce the same conditional mean velocity and hence the same physical first-order flux. However,
	\[
	\mathbb{E}|U_1|=\frac{a}{2},
	\qquad
	\mathbb{E}|U_2|=\frac{b}{2}.
	\]
	Their contributions to the integral defining $G_{ij}$ are therefore different on a set of positive line measure. Thus, the same $\lambda$ and $\mathbf{J}$ can correspond to different gross crossing rates.
\end{proof}

Propositions~\ref{prop:pairwise_transport_ambiguity} and \ref{prop:gross_crossing_nonidentifiability} concern different information losses. The first occurs when physical pairwise transport is inferred from density evolution. The second occurs when gross crossing traffic is inferred from first-order population transport.

\vspace*{-1mm}
\subsection{Closure for Unresolved Normal Mobility}\label{subsec:mobility_closure} % S3.2

We next consider the second information limitation while treating the physical first-order population flux as given.

\begin{Corollary} 
	\label{cor:coherent_normal_mobility}
	If $U_{ij}$ is almost surely nonnegative or almost surely nonpositive at almost every point of $\Gamma_{ij}$, then
	\begin{equation} % eq.19
		G_{ij}(t)
		=
		\int_{\Gamma_{ij}}
		\left|
		\mathbf{J}(\mathbf{x},t)
		\cdot
		\mathbf{n}_{ij}
		\right|
		\,\mathrm{d}\ell.
		\label{eq:coherent_closure}
	\end{equation}
	If the mean normal flux also has a common sign along $\Gamma_{ij}$, then
	\begin{equation} % eq.20
		G_{ij}(t)=|F_{ij}(t)|.
		\label{eq:coherent_net_gross}
	\end{equation}
\end{Corollary}

\begin{proof}
	Under the stated condition,
	\[
	\mathbb{E}
	\left[
	|U_{ij}|
	\mid
	\mathbf{x},t
	\right]
	=
	\left|
	\mathbb{E}
	\left[
	U_{ij}
	\mid
	\mathbf{x},t
	\right]
	\right|
	\]
	at almost every boundary point. Substitution into \eqref{eq:net_gross_flows} gives \eqref{eq:coherent_closure}. If the mean normal flux has a common sign along $\Gamma_{ij}$, the integral of its absolute value equals the absolute value of its integral, which gives \eqref{eq:coherent_net_gross}.
\end{proof}

Corollary~\ref{cor:coherent_normal_mobility} does not require identical UE velocities. It requires only that the conditional normal-velocity distribution at almost every boundary point be supported on one side of zero. When this condition is violated, an additional description of local velocity variation is required.

For counter-moving populations, assume that the conditional UE velocity has finite second moments wherever $\lambda(\mathbf{x},t)>0$. Define the conditional velocity covariance, normal mean, and normal variance respectively as
\begin{equation} % eq.21
	\begin{aligned}
		\bm{\Sigma}(\mathbf{x},t)
		&\triangleq
		\mathbb{E}
		\left[
		\left(
		\mathbf{U}-\mathbf{v}
		\right)
		\left(
		\mathbf{U}-\mathbf{v}
		\right)^{\mathsf T}
		\mid
		\mathbf{x},t
		\right],\\
		m_{ij}(\mathbf{x},t)
		&\triangleq
		\mathbf{v}(\mathbf{x},t)
		\cdot
		\mathbf{n}_{ij}
		=
		\frac{
			\mathbf{J}(\mathbf{x},t)
			\cdot
			\mathbf{n}_{ij}
		}{
			\lambda(\mathbf{x},t)
		},\\
		s_{ij}^{2}(\mathbf{x},t)
		&\triangleq
		\mathbf{n}_{ij}^{\mathsf T}
		\bm{\Sigma}(\mathbf{x},t)
		\mathbf{n}_{ij}.
	\end{aligned}
	\label{eq:normal_moments}
\end{equation}
These definitions apply wherever $\lambda(\mathbf{x},t)>0$, and zero-density points do not contribute to the boundary integrals.

We adopt the Gaussian moment closure
\begin{equation} % eq.22
	U_{ij}\mid\mathbf{x},t
	\sim
	\mathcal{N}
	\left(
	m_{ij}(\mathbf{x},t),
	s_{ij}^{2}(\mathbf{x},t)
	\right).
	\label{eq:normal_dispersion_model}
\end{equation}
This closure is fully specified by the prescribed first two moments and yields closed-form one-sided velocity moments. It is also the maximum-entropy distribution on $\mathbb{R}$ for fixed mean and variance. The closure is imposed only on the boundary-normal velocity and does not require the complete two-dimensional UE velocity distribution to be Gaussian.

Let $\varphi$ and $\Phi$ denote the probability density function and cumulative distribution function of a standard normal random variable. For $s_{ij}>0$, define
\begin{equation} % eq.23
	z_{ij}
	\triangleq
	\frac{m_{ij}}{s_{ij}}.
	\label{eq:normal_standardization}
\end{equation}
The one-sided Gaussian moments give
\begin{equation} % eq.24
	\begin{aligned}
		H_{ij}(t)
		&=
		\int_{\Gamma_{ij}}
		\lambda
		\left(
		s_{ij}\varphi(z_{ij})
		+
		m_{ij}\Phi(z_{ij})
		\right)
		\,\mathrm{d}\ell,\\
		G_{ij}(t)
		&=
		\int_{\Gamma_{ij}}
		\lambda
		\left(
		2s_{ij}\varphi(z_{ij})
		+
		m_{ij}
		\left(
		2\Phi(z_{ij})-1
		\right)
		\right)
		\,\mathrm{d}\ell.
	\end{aligned}
	\label{eq:gaussian_closure_rates}
\end{equation}
All quantities inside the integrals in \eqref{eq:gaussian_closure_rates} are evaluated at $(\mathbf{x},t)$. The reverse rate $H_{ji}$ follows by replacing $m_{ij}$ with $-m_{ij}$. Hence,
\begin{equation} % eq.25
	H_{ij}(t)-H_{ji}(t)=F_{ij}(t),
	\label{eq:gaussian_flux_preservation}
\end{equation}
and the closure preserves the prescribed first-order population transport. As $s_{ij}\rightarrow0$, \eqref{eq:gaussian_closure_rates} approaches the coherent result in Corollary~\ref{cor:coherent_normal_mobility}. When $m_{ij}=0$ and $s_{ij}>0$, the net normal flux is zero while the gross crossing rate remains positive.

The covariance $\bm{\Sigma}$ cannot be inferred from density snapshots alone. It must therefore be supplied by an auxiliary mobility statistic, a trajectory subsample, or a prescribed residual-mobility model. The Gaussian closure augments the first-order transport with this independently supplied second-order information.

As an alternative to the Gaussian closure, prescribed velocity moments and a bounded-speed assumption yield distribution-free crossing-rate intervals. The following result applies classical moment-bound principles~\cite{BertsimasPopescu2005} to directional crossing rates.

\begin{Proposition} 
	\label{prop:moment_crossing_interval}
	Suppose that $|U_{ij}|\leq v_{\max}$ almost surely on $\Gamma_{ij}$, where $v_{\max}>0$, and the conditional moments $m_{ij}$ and $s_{ij}^2$ in \eqref{eq:normal_moments} are given. Let $\eta_{ij}=m_{ij}^2+s_{ij}^2$ and define
	\begin{equation} % eq.26
		\begin{aligned}
			G_{ij}^{\mathrm L}&=\int_{\Gamma_{ij}}\lambda\max\left\{|m_{ij}|,\frac{\eta_{ij}}{v_{\max}}\right\}\,\mathrm d\ell,\\
			G_{ij}^{\mathrm U}&=\int_{\Gamma_{ij}}\lambda\sqrt{\eta_{ij}}\,\mathrm d\ell.
		\end{aligned}
		\label{eq:moment_gross_bounds}
	\end{equation}
	Then $G_{ij}^{\mathrm L}\leq G_{ij}\leq G_{ij}^{\mathrm U}$ and
	\begin{equation} % eq.27
		\frac{F_{ij}+G_{ij}^{\mathrm L}}{2}\leq H_{ij}\leq\frac{F_{ij}+G_{ij}^{\mathrm U}}{2}.
		\label{eq:moment_directional_bounds}
	\end{equation}
	Both endpoints are attainable within the class of boundary-local conditional velocity distributions satisfying the prescribed moments and support. Thus, the interval is sharp within this class.
\end{Proposition}

\begin{proof}
	Write $m=m_{ij}$ and $\eta=\eta_{ij}$ at a fixed boundary point. The inequalities
	\[
	|m|\leq\mathbb E|U_{ij}|,\qquad
	\eta\leq v_{\max}\mathbb E|U_{ij}|,\qquad
	\mathbb E|U_{ij}|\leq\sqrt{\eta}
	\]
	give \eqref{eq:moment_gross_bounds}. Equation \eqref{eq:directional_from_net_gross} then yields \eqref{eq:moment_directional_bounds}. For $\eta>0$, the upper bound is attained by $U_{ij}\in\{-\sqrt{\eta},\sqrt{\eta}\}$ with probabilities $(1-m/\sqrt{\eta})/2$ and $(1+m/\sqrt{\eta})/2$. For $0<\eta\leq v_{\max}|m|$, the lower bound is attained by $U_{ij}=\operatorname{sgn}(m)\eta/|m|$ with probability $m^2/\eta$ and zero otherwise. For $\eta\geq v_{\max}|m|$, it is attained by $U_{ij}\in\{-v_{\max},0,v_{\max}\}$ with probabilities $(\eta-v_{\max}m)/(2v_{\max}^2)$, $1-\eta/v_{\max}^2$, and $(\eta+v_{\max}m)/(2v_{\max}^2)$, respectively. When $\eta=0$, both bounds vanish. Applying these constructions along the boundary establishes attainability of the integrated endpoints.
\end{proof}

The interval width is $(G_{ij}^{\mathrm U}-G_{ij}^{\mathrm L})/2$. It vanishes when $s_{ij}=0$ or $\eta_{ij}=v_{\max}^2$ almost everywhere on the boundary. The bounds do not require a prescribed velocity distribution.

\vspace*{-1mm}
\subsection{Stationary Isotropic Limit}\label{subsec:stationary_limit} % S3.3

The stationary isotropic case provides a consistency check with the classical stochastic-geometry result and does not use the Gaussian closure. Consider an infinite plane in which BSs form a homogeneous Poisson point process with intensity $\mu_{\mathrm{BS}}$. Let the UE density be constant, denoted by $\lambda_{\mathrm{u}}$, and let each UE move with constant speed $v_0$ in an independently and uniformly distributed direction. Nearest-BS association then induces a stationary Poisson--Voronoi tessellation.

For an arbitrary cell boundary, let $\Theta$ denote the angle between the UE velocity and a selected boundary normal. Isotropy gives
$\Theta\sim\operatorname{Unif}[0,2\pi)$.
Let $\mathcal{L}_{\mathrm e}$ denote the mean Voronoi edge length per unit area when each geometric edge is counted once. For a homogeneous Poisson--Voronoi tessellation,
$\mathcal{L}_{\mathrm e}=2\sqrt{\mu_{\mathrm{BS}}}$.
The mean absolute normal speed and the resulting crossing rate per UE are therefore given respectively by
\begin{equation} % eq.28
	\begin{aligned}
		\mathbb{E}
		\left[
		|v_0\cos\Theta|
		\right]
		&=
		\frac{2v_0}{\pi},\\
		\overline h_{\mathrm{UE}}
		&=
		\frac{
			\lambda_{\mathrm{u}}
			\mathbb{E}[|v_0\cos\Theta|]
			\mathcal{L}_{\mathrm e}
		}{
			\lambda_{\mathrm{u}}
		}
		=
		\frac{
			4v_0\sqrt{\mu_{\mathrm{BS}}}
		}{
			\pi
		}.
	\end{aligned}
	\label{eq:stationary_limit}
\end{equation}
Equation~\eqref{eq:stationary_limit} is the classical mean HO rate under isotropic mobility and nearest-BS association in a homogeneous Poisson cellular network \cite{Lin2013Mobility}.

This special case also illustrates the information absent from the first-order population field. Isotropy gives
$\mathbb{E}[\mathbf{U}]=\mathbf{0}$ and hence $\mathbf{J}=\mathbf{0}$, even though
$\mathbb{E}|U_{ij}|>0$.
Because the density is stationary, the minimum-energy F-STSG reconstruction also gives
$\widehat{\mathbf{J}}=\mathbf{0}$.
The nonzero crossing rate in \eqref{eq:stationary_limit} therefore arises from velocity variation that is not represented by the first-order population flux.

\vspace*{-1mm}
%========================IV=================================
\section{Handover Estimation from Density Snapshots}\label{sec:estimation} % S4

The preceding analysis separates HO estimation from aggregate observations into two steps. Density evolution is first used to reconstruct an admissible first-order population transport field, after which directional crossing rates are obtained under a specified local mobility closure. This section implements these steps for discrete density snapshots. The F-STSG reconstruction selects the minimum-energy transport field consistent with the observed population evolution. Proposition~\ref{prop:pairwise_transport_ambiguity} and Corollary~\ref{cor:potential_exact_recovery} characterize when this field can represent the physical pairwise transport, while the local mobility closure determines its conversion to directional crossing rates.

A density snapshot is a spatial estimate of the UE population over the service region at a given observation time. Such a field may be constructed from gridded positioning estimates, spatially aggregated measurement reports, or other location-aware population measurements~\cite{yuan2025constructing}. Cell-level UE counts provide only population integrals over the serving regions and therefore require an additional spatial interpolation or density model before continuous reconstruction can be applied. The effect of density observation errors is examined in Section~\ref{subsec:resolution_results}.

\vspace*{-1mm}
\subsection{Population Flux Reconstruction}\label{subsec:flux_reconstruction} % S4.1

Suppose that density snapshots are available at uniformly spaced times
$t_k=t_0+k\Delta t$, where $k=0,\ldots,K$ and $\Delta t>0$.
Let $\lambda_k(\mathbf{x})$ denote the density snapshot used for reconstruction at time $t_k$. Under ideal observation,
$\lambda_k(\mathbf{x})\! =\! \lambda(\mathbf{x},t_k)$, and perturbed snapshots are considered separately in the numerical evaluation. For each interior snapshot $k=1,\ldots,K-1$, the temporal derivative is approximated by the centered difference
\begin{equation} % eq.29
	\widetilde q_k(\mathbf{x})
	=
	\frac{
		\lambda_{k+1}(\mathbf{x})
		-
		\lambda_{k-1}(\mathbf{x})
	}{
		2\Delta t
	}.
	\label{eq:snapshot_derivative}
\end{equation}
For sufficiently smooth density fields, \eqref{eq:snapshot_derivative} is second-order accurate. In the midpoint-based numerical experiments, we instead use $(\lambda_{k+1}-\lambda_k)/\Delta t$ with the midpoint density $(\lambda_{k+1}+\lambda_k)/2$, retaining second-order temporal accuracy at the interval midpoint. For the closed system in Section~\ref{sec:system_model}, the weighted Neumann problem requires the temporal derivative to have zero spatial integral. This condition follows from constant total population, while small numerical inconsistencies are corrected in the corrected temporal derivative $q_k$ given by
\begin{equation} % eq.30
	q_k(\mathbf{x})
	=
	\widetilde q_k(\mathbf{x})
	\! -\!
	\frac{1}{|\Omega|}
	\! \int_{\Omega}\!
	\widetilde q_k(\mathbf{y})
	\,\mathrm{d}\mathbf{y}
	\mbox{ with }
	\! \int_{\Omega}\!
	q_k(\mathbf{x})
	\,\mathrm{d}\mathbf{x}
	=
	0,
	\label{eq:conservative_derivative}
\end{equation}
where $|\Omega|$ denotes the area of the service region. This correction is intended only for small compatibility errors. A physical change in the total UE population must instead be represented through an open boundary or a source term and should not be removed by \eqref{eq:conservative_derivative}.

Using $\lambda_k$ and $q_k$, the F-STSG reconstruction at time $t_k$ is obtained from
\begin{equation} % eq.31
	\begin{aligned}
		\nabla\!\cdot
		\left(
		\lambda_k(\mathbf{x})
		\nabla\phi_k(\mathbf{x})
		\right)
		&=
		q_k(\mathbf{x}),
		\qquad
		\mathbf{x}\in\Omega,\\
		\lambda_k(\mathbf{x})
		\nabla\phi_k(\mathbf{x})
		\cdot
		\mathbf{n}_{\Omega}(\mathbf{x})
		&=
		0,
		\qquad
		\mathbf{x}\in\partial\Omega,\\
		\widehat{\mathbf{J}}_k(\mathbf{x})
		&=
		-\lambda_k(\mathbf{x})
		\nabla\phi_k(\mathbf{x}).
	\end{aligned}
	\label{eq:snapshot_fstsg}
\end{equation}
The potential is determined up to an additive constant, which does not affect
$\widehat{\mathbf{J}}_k$. We assume that $\lambda_k$ is sufficiently regular and bounded away from zero on the reconstruction domain, as required by the weighted elliptic formulation. In numerical implementation, zero-density regions can be excluded from the active reconstruction support or handled through the regularization adopted by the F-STSG solver.

After spatial discretization, \eqref{eq:snapshot_fstsg} yields a sparse linear system with one potential unknown per active grid point. For the local finite-volume discretization used here, the number of nonzero matrix entries grows linearly with the number of active grid points. This sparsity property does not imply linear solution time, which also depends on matrix conditioning and iterative convergence. Reconstructions at different snapshot times are independent and can be processed in parallel.

The reconstructed signed flow across a common cell boundary is
\begin{equation} % eq.32
	\widehat F_{ij,k}
	=
	\int_{\Gamma_{ij}}
	\widehat{\mathbf{J}}_k(\mathbf{x})
	\cdot
	\mathbf{n}_{ij}
	\,\mathrm{d}\ell.
	\label{eq:snapshot_pairwise_flux}
\end{equation}
Because \eqref{eq:snapshot_fstsg} satisfies the population balance associated with $q_k$, the reconstructed cell flows obey
\begin{equation} % eq.33
	\sum_{j\in\mathcal{N}_i}
	\widehat F_{ij,k}
	=
	-
	\int_{\mathcal{V}_i}
	q_k(\mathbf{x})
	\,\mathrm{d}\mathbf{x}.
	\label{eq:snapshot_cell_balance}
\end{equation}
Equation~\eqref{eq:snapshot_cell_balance} provides both a population-consistency relation and a numerical check on the boundary integration. It ensures that the reconstructed pairwise flows reproduce the cell population change implied by the density snapshots, but it does not imply
$\widehat F_{ij,k}=F_{ij}(t_k)$
for every neighboring pair. Such equality depends on the identifiability of the physical transport from density evolution. In particular, Corollary~\ref{cor:potential_exact_recovery} identifies a potential-transport class for which the continuous F-STSG reconstruction coincides with the physical first-order flux when the exact density derivative is available.

\vspace*{-1mm}
\subsection{Directional Handover Estimation}\label{subsec:handover_estimation} % S4.2

The reconstructed field in \eqref{eq:snapshot_fstsg} determines signed first-order transport consistent with the observed density evolution. The remaining step is to map this transport to directional mobility-induced crossing rates. As established in Section~\ref{subsec:mobility_closure}, this mapping depends on the unresolved local normal-velocity distribution and is separate from the transport reconstruction.

When local normal motion is coherent, the reconstructed directional rate is
\begin{equation} % eq.34
	\widehat H_{ij,k}^{\mathrm{coh}}
	=
	\int_{\Gamma_{ij}}
	\left[
	\widehat{\mathbf{J}}_k(\mathbf{x})
	\cdot
	\mathbf{n}_{ij}
	\right]_{+}
	\,\mathrm{d}\ell.
	\label{eq:snapshot_coherent_rate}
\end{equation}
The reverse rate is obtained by replacing $\mathbf{n}_{ij}$ with
$-\mathbf{n}_{ij}$, and hence
\begin{equation} % eq.35
	\widehat H_{ij,k}^{\mathrm{coh}}
	-
	\widehat H_{ji,k}^{\mathrm{coh}}
	=
	\widehat F_{ij,k}.
	\label{eq:snapshot_coherent_net}
\end{equation}

The effect of first-order transport reconstruction error on the directional estimate is bounded directly.

\begin{Proposition} 
	\label{prop:coherent_ho_error_bound}
	Suppose that the physical normal mobility is coherent on $\Gamma_{ij}$ at time $t_k$. Then
	\begin{equation} % eq.36
		\left|
		\widehat H_{ij,k}^{\mathrm{coh}}
		-
		H_{ij}(t_k)
		\right|
		\!\leq\!
		\int_{\Gamma_{ij}}
		\left|
		\left(
		\widehat{\mathbf{J}}_k(\mathbf{x})
		\!-\!
		\mathbf{J}(\mathbf{x},t_k)
		\right)
		\cdot
		\mathbf{n}_{ij}
		\right|
		\,\mathrm{d}\ell.
		\label{eq:coherent_estimation_bound}
	\end{equation}
	Consequently, convergence of the reconstructed boundary-normal transport in
	$L^1(\Gamma_{ij})$
	implies convergence of the corresponding coherent directional crossing estimate.
\end{Proposition}

\begin{proof}
	Under coherent normal mobility, Corollary~\ref{cor:coherent_normal_mobility} gives
	\[
	H_{ij}(t_k)
	=
	\int_{\Gamma_{ij}}
	\left[
	\mathbf{J}(\mathbf{x},t_k)
	\cdot
	\mathbf{n}_{ij}
	\right]_+
	\,\mathrm{d}\ell.
	\]
	The positive-part mapping is nonexpansive,
	and $|[a]_+-[b]_+|\leq|a-b|$.
	Applying this inequality pointwise on $\Gamma_{ij}$ and integrating gives \eqref{eq:coherent_estimation_bound}.
\end{proof}

Proposition~\ref{prop:coherent_ho_error_bound} bounds coherent directional estimation error by the difference between reconstructed and physical boundary-normal fluxes. This difference includes both snapshot reconstruction error and, when present, canonical-to-physical transport mismatch.

When local velocity dispersion cannot be neglected, let
$\widehat{\bm{\Sigma}}_k(\mathbf{x})$
denote an independently supplied estimate of the local velocity covariance. For locations where $\lambda_k(\mathbf{x})>0$, define
\begin{equation} % eq.37
	\widehat m_{ij,k}(\mathbf{x})
	=
	\frac{
		\widehat{\mathbf{J}}_k(\mathbf{x})
		\cdot
		\mathbf{n}_{ij}
	}{
		\lambda_k(\mathbf{x})
	},
	\qquad
	\widehat s_{ij,k}^{2}(\mathbf{x})
	=
	\mathbf{n}_{ij}^{\mathsf T}
	\widehat{\bm{\Sigma}}_k(\mathbf{x})
	\mathbf{n}_{ij}.
	\label{eq:snapshot_normal_statistics}
\end{equation}
Points with $\lambda_k\! =\! 0$ make no contribution to the crossing integrals.

For $\widehat s_{ij,k}>0$, let
\begin{equation} % eq.38
	\widehat z_{ij,k}
	=
	\frac{
		\widehat m_{ij,k}
	}{
		\widehat s_{ij,k}
	}.
	\label{eq:snapshot_normal_standardization}
\end{equation}
Using the Gaussian moment closure in \eqref{eq:normal_dispersion_model}, the dispersion-aware directional estimate is
\begin{equation} % eq.39
	\widehat H_{ij,k}^{\mathrm{disp}}
	=
	\int_{\Gamma_{ij}}
	\lambda_k
	\left(
	\widehat s_{ij,k}
	\varphi(\widehat z_{ij,k})
	+
	\widehat m_{ij,k}
	\Phi(\widehat z_{ij,k})
	\right)
	\,\mathrm{d}\ell.
	\label{eq:snapshot_dispersion_rate}
\end{equation}
All quantities inside \eqref{eq:snapshot_dispersion_rate} are evaluated at $\mathbf{x}$ and time $t_k$. The covariance field is additional mobility information and cannot be inferred from density snapshots alone. It may be obtained from an auxiliary mobility measurement, a trajectory subsample, or a prescribed residual-mobility model.

The dispersion-aware estimate also admits an error bound that separates transport reconstruction from the information absent in density evolution and first-order mobility. Define the Gaussian positive-part moment
\begin{equation} % eq.40
\mathcal K(f,a)=
\begin{cases}
 a\varphi(f/a)+f\Phi(f/a), & a>0,\\
 [f]_+, & a=0.
\end{cases}
\label{eq:gaussian_positive_part_kernel}
\end{equation}
At time $t_k$, let $f=\mathbf J\cdot\mathbf n_{ij}$,
$f^{\mathrm c}=\widehat{\mathbf J}\cdot\mathbf n_{ij}$,
$\widehat f_k=\widehat{\mathbf J}_k\cdot\mathbf n_{ij}$,
$a=\lambda s_{ij}$, and $\widehat a_k=\lambda_k\widehat s_{ij,k}$.
Here $\widehat{\mathbf J}$ is the canonical field obtained from the exact density evolution in \eqref{eq:fstsg_reconstruction}, whereas $\widehat{\mathbf J}_k$ is reconstructed from snapshots. Write $r_{ij}=\lambda\mathbb E[[U_{ij}]_+\mid\mathbf x,t_k]$. All local quantities are evaluated on $\Gamma_{ij}$ at $t_k$. Then
\begin{equation} % eq.41
\bigl|\widehat H_{ij,k}^{\mathrm{disp}}-H_{ij}(t_k)\bigr|
\le E_{\mathrm{rec}}+E_{\mathrm{id}}+E_{\mathrm{stat}}+E_{\mathrm{cl}},
\label{eq:directional_error_decomposition}
\end{equation}
where the boundary-integrated terms are
\begin{equation} % eq.42
\begin{aligned}
E_{\mathrm{rec}}&=\int_{\Gamma_{ij}}|\widehat f_k-f^{\mathrm c}|\,\mathrm d\ell,\\
E_{\mathrm{id}}&=\int_{\Gamma_{ij}}|f^{\mathrm c}-f|\,\mathrm d\ell,\\
E_{\mathrm{stat}}&=\frac{1}{\sqrt{2\pi}}\int_{\Gamma_{ij}}|\widehat a_k-a|\,\mathrm d\ell,\\
E_{\mathrm{cl}}&=\int_{\Gamma_{ij}}|\mathcal K(f,a)-r_{ij}|\,\mathrm d\ell.
\end{aligned}
\label{eq:directional_error_terms}
\end{equation}
Indeed, $|\partial_f\mathcal K|\leq1$ and $|\partial_a\mathcal K|\leq1/\sqrt{2\pi}$ for $a>0$, and the same Lipschitz bounds extend to $a=0$ by continuity. Equation~\eqref{eq:directional_error_decomposition} follows by the triangle inequality. The terms quantify bounds associated with snapshot reconstruction, canonical-to-physical flux mismatch, auxiliary statistics, and distributional closure, respectively, and they are not additive components of the realized estimation error.

For each snapshot, the estimated neighboring-cell rates form an $M\times M$ geometric crossing-rate matrix
\begin{equation} % eq.43
	\widehat{\mathbf{H}}_k
	=
	\left[
	\widehat H_{ij,k}
	\right].
\end{equation}
Here, $\widehat H_{ii,k}=0$ and entries corresponding to nonneighboring cell pairs are zero. The entry $\widehat H_{ij,k}$ is obtained from either the coherent closure in \eqref{eq:snapshot_coherent_rate} or the dispersion-aware closure in \eqref{eq:snapshot_dispersion_rate}, according to the adopted local mobility model.

The overall computational chain is
\begin{equation} % eq.44
	\begin{aligned}
		\{\lambda_{k-1},\lambda_k,\lambda_{k+1}\}
		&\longrightarrow q_k
		\longrightarrow \widehat{\mathbf{J}}_k,\\
		\widehat{\mathbf{J}}_k
		&\longrightarrow
		\left\{
		\widehat F_{ij,k},
		\ \widehat H_{ij,k}\ \text{(with closure)}
		\right\}.
	\end{aligned}
	\label{eq:estimation_chain}
\end{equation}
The first part of \eqref{eq:estimation_chain} reconstructs an F-STSG transport field consistent with the observed density evolution, whereas the final mapping determines directional crossing rates under the adopted local mobility closure. Under coherent normal mobility, no additional velocity statistic is required, and Proposition~\ref{prop:coherent_ho_error_bound} relates the directional estimation error directly to the boundary-normal transport error. Under unresolved counterflow, the dispersion-aware estimate additionally requires the local normal-velocity variance. Agreement between $\widehat{\mathbf{H}}_k$ and the physical mobility-induced crossing-rate matrix therefore depends separately on transport reconstruction and on the adequacy of the local mobility closure.

\vspace*{-1mm}
%======================== V ================================
\section{Numerical Results}\label{sec:numerical_results} % S5

This section evaluates the kinetic crossing model, its information limits, and density-snapshot reconstruction. Controlled mobility fields test the conditions in Section~\ref{sec:ho_characterization}, while an independent vehicular simulation provides external validation. Unless otherwise stated, association follows the nearest BS. In synthetic experiments, density, flux, and crossing rates are normalized by the total UE population $N_{\mathrm{UE}}$.

For a directed neighboring-cell quantity $Q_{ij}$, where $Q$ denotes $H$ or $F$, we use the normalized pairwise error
\begin{equation} % eq.45
	\epsilon_Q
	=
	\frac{
		\sum_{(i,j)\in\mathcal E_{\mathrm d}}
		|Q_{ij}^{\mathrm{est}}-Q_{ij}^{\mathrm{ref}}|
	}{
		\sum_{(i,j)\in\mathcal E_{\mathrm d}}
		|Q_{ij}^{\mathrm{ref}}|
	},
	\qquad Q\in\{H,F\},
	\label{eq:numerical_pairwise_error}
\end{equation}
where $\mathcal E_{\mathrm d}=\{(i,j):j\in\mathcal N_i\}$ is the set of ordered neighboring pairs. In the controlled experiments, the reference quantities are calculated from the prescribed physical mobility field. 
The flux-field error is $\epsilon_J=\|\widehat{\mathbf{J}}-\mathbf{J}\|_{L^2(\Omega)}/\|\mathbf{J}\|_{L^2(\Omega)}$, evaluated at the corresponding reconstruction time. 
We also evaluate cell-level population conservation using the normalized balance residual $\epsilon_{\mathrm{bal}}$,
 	defined as the accumulated absolute mismatch between the cell-population change and the net outgoing flow,
 	normalized by the corresponding absolute magnitudes. We distinguish $\epsilon_{\mathrm{bal}}^{\mathrm{phys}}$, computed from the physical flux, from $\epsilon_{\mathrm{bal}}^{\mathrm{F-STSG}}$, computed from the reconstructed flux. Both use the same cell-population changes but their respective boundary fluxes, and measure conservation consistency rather than pairwise reconstruction accuracy. 
For external evaluation, $Q_{ij}$ denotes window-integrated counts. The window-level net-flow error is denoted as $\epsilon_F^{(w)}$, and the paired baseline difference is:
\begin{equation} % eq.46
	\Delta\epsilon_F^{(w)}
	=
	\epsilon_{F,\mathrm{baseline}}^{(w)}
	-
	\epsilon_{F,\mathrm{F\text{-}STSG}}^{(w)},
	\label{eq:paired_net_error}
\end{equation}
where $\epsilon_{F,\mathrm{baseline}}^{(w)}$ and $\epsilon_{F,\mathrm{F\text{-}STSG}}^{(w)}$ are the normalized pairwise net-flow errors of the baseline and F-STSG, respectively, evaluated on the same directed cell pairs in window $w$. Mean paired differences across windows are distinct from the pooled normalized error over all directed pairs.

\vspace*{-1mm}
\subsection{Nonstationary Coherent Mobility}\label{subsec:coherent_results} % S5.1

We first consider a potential-transport field satisfying the exact-recovery condition in Corollary~\ref{cor:potential_exact_recovery} and the coherent boundary-normal mobility condition in Corollary~\ref{cor:coherent_normal_mobility}. The simulation region is a $2~\mathrm{km}\times2~\mathrm{km}$ square containing 100 fixed BSs and 10,000 UEs. The initial population consists of a uniform background and a Gaussian hotspot with mixture weight $0.65$ and standard deviation $180$~m. A horizontal potential flow with maximum speed $15$~m/s drives the population evolution under a no-flux outer boundary condition. Microscopic trajectories are sampled every $0.05$~s over 30 independent UE realizations. The F-STSG reconstruction uses noise-free analytical density snapshots at $1$~s intervals on a grid with $20$~m spacing. The physical reference is obtained directly from the prescribed mobility field, whereas the F-STSG estimate uses density snapshots without access to the physical velocity field.

\begin{figure}[t]
	\centering
	\includegraphics[width=0.99\columnwidth]{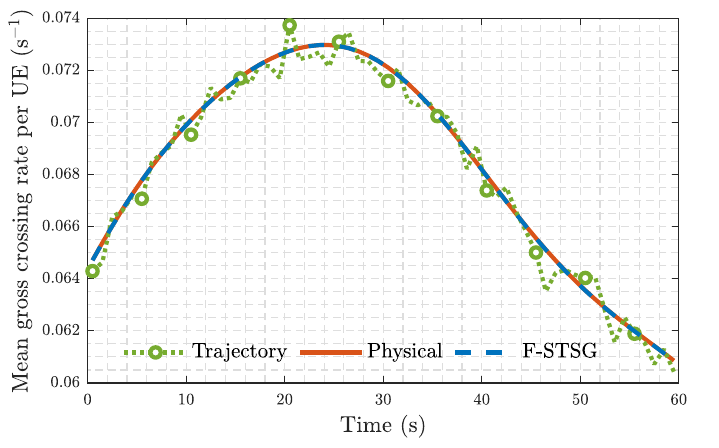}
	\vspace*{-9mm}
	\caption{Mean gross crossing rate per UE under nonstationary coherent mobility, comparing microscopic trajectory counts, physical boundary integration, and F-STSG reconstruction.}
	\label{fig:coherent_dynamic} % Fig.2
	\vspace*{-5mm}
\end{figure}

Fig.~\ref{fig:coherent_dynamic} compares the three gross crossing-rate estimates. The rate first rises from $0.064$ to $0.073~\mathrm{s}^{-1}$ and then falls to $0.061~\mathrm{s}^{-1}$ as the UE distribution moves across fixed serving regions. The physical and F-STSG curves nearly coincide, whereas trajectory counts fluctuate with finite sampling. Relative to the physical boundary reference, trajectories yield $\epsilon_H=\epsilon_F=9.60\times10^{-3}$, versus $7.32\times10^{-4}$ for F-STSG, and the cell-balance residual is $1.19\times10^{-2}$. The smaller F-STSG error reflects analytical density inputs rather than a general advantage over trajectory observations. The results are consistent with the potential-transport recovery condition and Proposition~\ref{prop:coherent_ho_error_bound}.

\begin{figure}[b]
	\vspace*{-5mm}
	\centering
	\includegraphics[width=0.9\columnwidth]{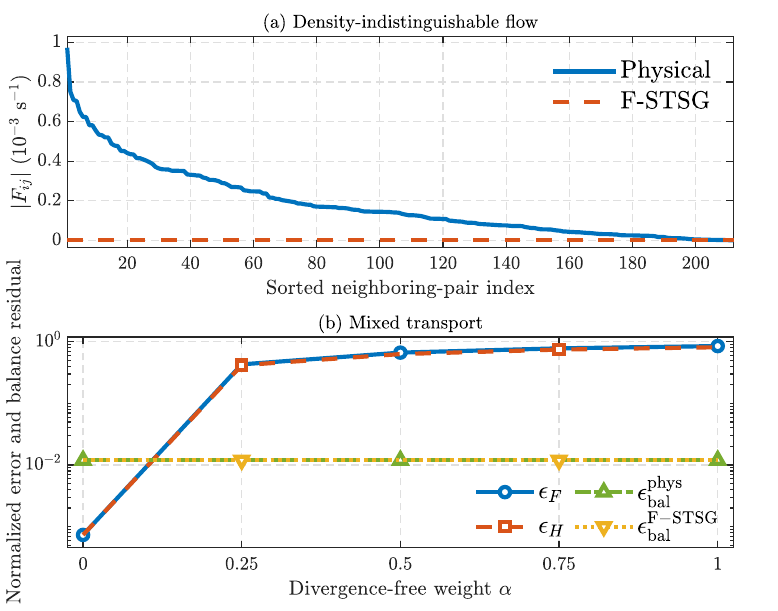}
	\vspace*{-8mm}
\caption{Effect of divergence-free transport on pairwise flow reconstruction. (a) Physical and F-STSG pairwise net-flow magnitudes under stationary density. (b) Net-flow and directional crossing errors versus the divergence-free weight $\alpha$, compared with the physical and reconstructed cell-balance residuals. The balance residuals measure conservation consistency, not pairwise reconstruction accuracy. Flows in (a) are normalized per UE.}
	\label{fig:circulation_limit} % Fig.3
	\vspace*{-1mm}
\end{figure}

\vspace*{-2mm}
\subsection{Information Limits of First-Order Population Fields}
\label{subsec:information_limits_results} % S5.2

We next examine two information limits of first-order population fields: the nonuniqueness of physical pairwise transport given density evolution and the inability of first-order transport to determine gross crossing rates.

\begin{figure}[!b]
\vspace*{-6mm}
	\centering
	\includegraphics[width=0.85\columnwidth]{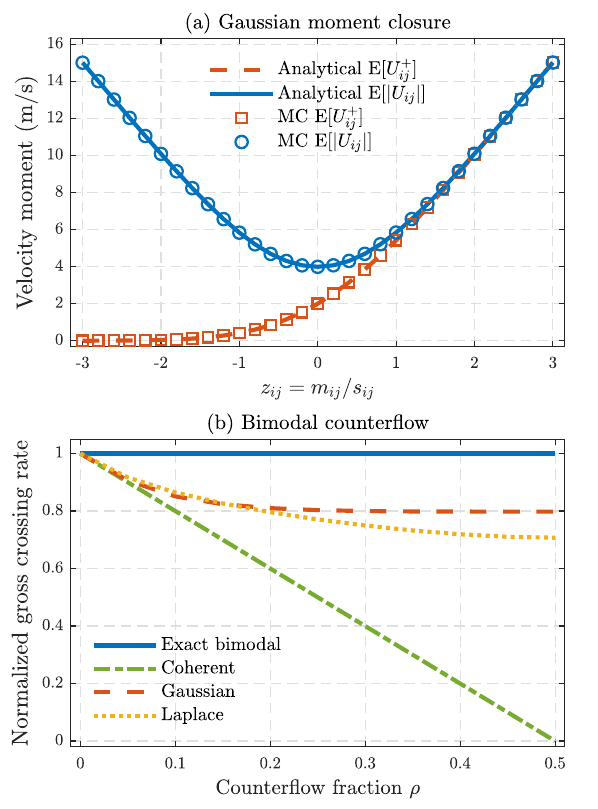}
	\vspace*{-4mm}
	\caption{Normal-velocity closure and counterflow. (a) Analytical Gaussian positive-part and absolute-velocity moments compared with Monte Carlo estimates. (b) Normalized gross crossing rates under bimodal counterflow, comparing the exact result with coherent, Gaussian, and Laplace closures. The Gaussian and Laplace closures are matched to the mean and variance of the bimodal distribution.}
	\label{fig:counterflow_closure} % Fig.4
	\vspace*{-1mm}
\end{figure}

To illustrate the first limit, we consider a disk of radius $1000$~m with 80 BSs and 100,000 uniformly distributed UEs undergoing rigid rotation at $0.01~\mathrm{s}^{-1}$. Over $60$~s, the divergence-free transport preserves the density while producing inter-cell crossings. Trajectory-derived pairwise flows agree with physical boundary integration to within a normalized error of $2.11\times10^{-2}$. In Fig.~\ref{fig:circulation_limit}(a), the physical pairwise flows remain nonzero, whereas the canonical F-STSG flux vanishes under stationary density, illustrating Proposition~\ref{prop:pairwise_transport_ambiguity}.

\begin{figure}[!b]
\vspace*{-6mm}
	\centering
	\includegraphics[width=0.85\columnwidth]{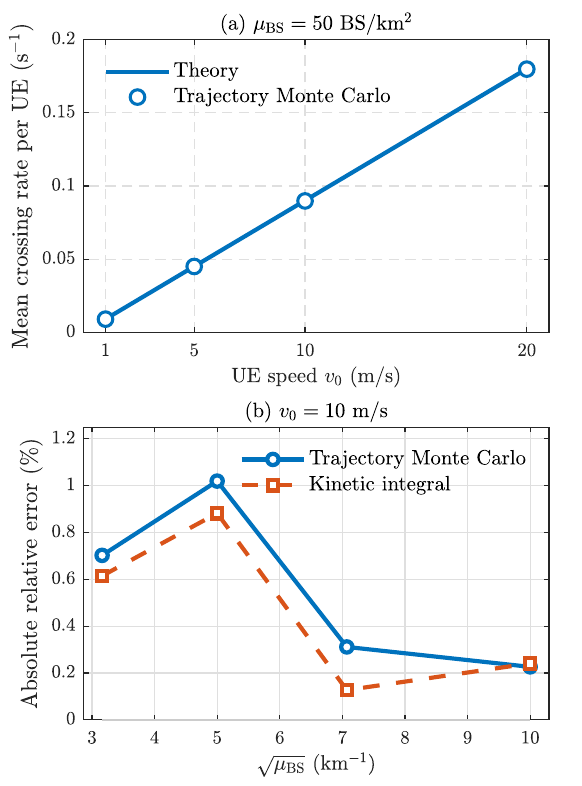}
	\vspace*{-4mm}
	\caption{Stationary Poisson--Voronoi crossing-rate limit. (a) Mean crossing rate per UE versus speed at $\mu_{\mathrm{BS}}=50$~BS/km$^2$, comparing the theoretical result with trajectory Monte Carlo estimates. (b) Absolute relative errors of trajectory Monte Carlo and kinetic boundary integration with respect to the theoretical rate across BS densities at $v_0=10$~m/s.}
	\label{fig:stationary_limit} % Fig.5
	\vspace*{-1mm}
\end{figure}

We further examine how unobserved divergence-free transport affects pairwise reconstruction. Starting from the potential flow in Section~\ref{subsec:coherent_results}, we add a divergence-free component generated by the stream function $\chi(x,y)=\sin^2(\pi x/L)\sin^2(\pi y/L)$, where $L=2$~km. Its $L^2$ norm is matched to that of the potential flux before scaling by $\alpha$. Because this component has zero divergence and is tangent to the outer boundary, varying $\alpha$ changes the physical pairwise flows without changing the density evolution or the F-STSG reconstruction.
Fig.~\ref{fig:circulation_limit}(b) shows that increasing $\alpha$ from zero to one raises the normalized net-flow and directional crossing errors from approximately $7.32\times10^{-4}$ to $0.854$ and $0.817$, respectively. In contrast, the physical and F-STSG cell-balance residuals remain near $1.2\times10^{-2}$, mainly due to differences between grid-based population integration and Voronoi boundary flux integration. Thus, both fields can satisfy nearly the same population balance while having substantially different pairwise flows, as characterized in Proposition~\ref{prop:pairwise_transport_ambiguity}.

For the second information limit, Fig.~\ref{fig:counterflow_closure}(a) compares the Gaussian moments in \eqref{eq:gaussian_closure_rates} with Monte Carlo estimates using $10^6$ independent samples per point, $s_{ij}=5$~m/s, and $z_{ij}=m_{ij}/s_{ij}\in[-3,3]$. The analytical and sampled positive-part and absolute-velocity moments agree closely.

Fig.~\ref{fig:counterflow_closure}(b) considers velocities of $+10$ and $-10$~m/s, with reverse-moving fraction $\rho\in[0,0.5]$. The exact normalized gross rate remains one, whereas the coherent estimate falls to zero as opposite flows cancel. At $\rho=0.5$, Gaussian and Laplace closures matched to the first two velocity moments yield $0.798$ and $0.707$, respectively. For this fixed-speed case, $\eta_{ij}=v_{\max}^2$ with $v_{\max}=10$~m/s, and therefore Proposition~\ref{prop:moment_crossing_interval} recovers the exact gross crossing rate.

\begin{figure*}[t]
\vspace*{-1mm}
	\centering
	\includegraphics[width=0.99\textwidth]{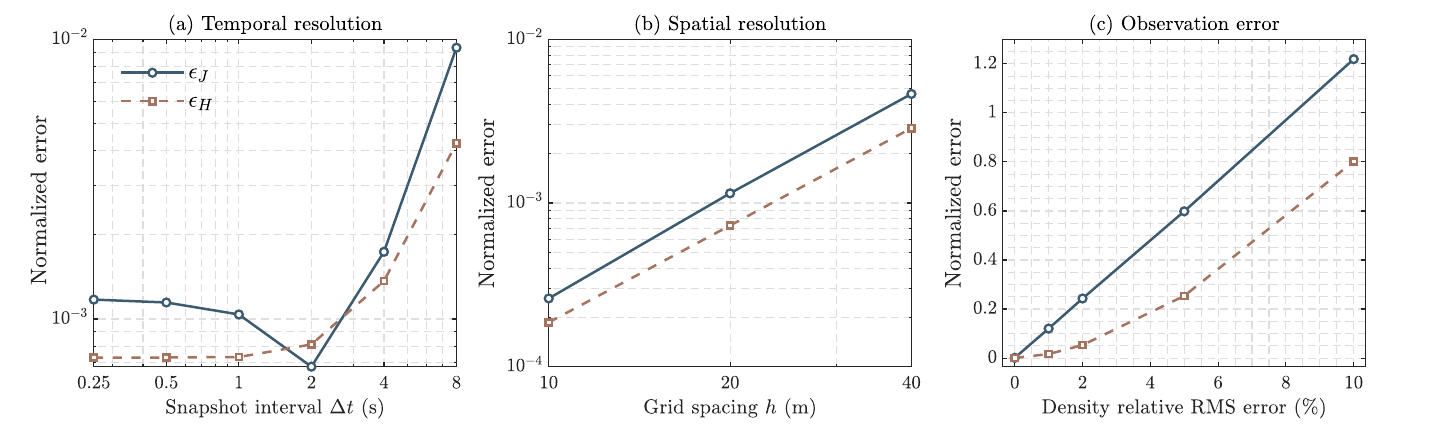}
	\vspace*{-5mm}
	\caption{Sensitivity of density-based reconstruction to observation resolution and accuracy. (a) Normalized flux and directional crossing errors versus snapshot interval $\Delta t$ at $h=20$~m. (b) Reconstruction errors versus spatial grid spacing $h$ at $\Delta t=0.5$~s. (c) Reconstruction errors versus relative root-mean-square (RMS) density error at $\Delta t=1$~s.}
	\label{fig:snapshot_interval} % Fig.6
	\vspace*{-5mm}
\end{figure*}

\vspace*{-1mm}
\subsection{Stationary Isotropic Limit}\label{subsec:stationary_results} % S5.3

We examine the stationary Poisson--Voronoi limit derived in Section~\ref{subsec:stationary_limit}. BS locations follow a homogeneous Poisson point process, and UEs move at constant speed in independent, uniformly distributed directions. Statistics are collected in a $2~\mathrm{km}\times2~\mathrm{km}$ central region with a $1500$~m guard margin. Each operating point uses 5000 UEs and 30 independent BS realizations. The trajectory Monte Carlo method directly counts intersections between UE trajectories and Voronoi boundaries, whereas the kinetic integral evaluates boundary-crossing rates using the realized Voronoi geometry and the prescribed velocity distribution. Both methods are evaluated against the theoretical result in \eqref{eq:stationary_limit}.

Fig.~\ref{fig:stationary_limit}(a) confirms the predicted linear dependence on UE speed. Across the BS densities in Fig.~\ref{fig:stationary_limit}(b), the maximum trajectory-to-theory deviation is $1.02\%$, while the largest trajectory-to-kinetic-integral difference is $0.184\%$ of the theoretical rate. Both results are consistent with the square-root BS-density scaling in \eqref{eq:stationary_limit}.

\vspace*{-1mm}
\subsection{Snapshot Resolution and Observation Error}\label{subsec:resolution_results} % S5.4

We next evaluate the sensitivity of density-based reconstruction to temporal resolution, spatial discretization, and observation errors under the coherent nonstationary model. The temporal experiment uses analytical density snapshots with intervals $\Delta t$ ranging from $0.25$ to $8$~s at a fixed grid spacing of $h=20$~m. The spatial experiment varies $h\in\{10,20,40\}$~m at $\Delta t=0.5$~s. To assess observation sensitivity, spatially correlated perturbations with a correlation length of $100$~m are added to the analytical density snapshots, with 30 independent realizations at each noise level.

In Fig.~\ref{fig:snapshot_interval}(a), errors change little below $\Delta t=1$--$2$~s: for $0.25\leq\Delta t\leq2$~s, $\epsilon_J<1.18\times10^{-3}$ and $\epsilon_H<8.12\times10^{-4}$. At $\Delta t=8$~s they increase to $9.33\times10^{-3}$ and $4.25\times10^{-3}$, respectively, indicating diminishing returns from finer temporal sampling at fixed spatial resolution.

Fig.~\ref{fig:snapshot_interval}(b) shows that reducing $h$ from $40$ to $10$~m decreases $\epsilon_J$ from $4.64\times10^{-3}$ to $2.61\times10^{-4}$ and $\epsilon_H$ from $2.87\times10^{-3}$ to $1.86\times10^{-4}$. For $h=40$, $20$, and $10$~m (2,500, 10,000, and 40,000 unknowns), mean elliptic-solve times are $8.37$, $38.75$, and $457.52$~ms per snapshot. Timings use a single-threaded Jacobi-preconditioned conjugate-gradient solver on an AMD Ryzen 5 9600X and exclude density acquisition and boundary integration.

Fig.~\ref{fig:snapshot_interval}(c) shows stronger sensitivity to density-observation error: at $1\%$, $2\%$, and $10\%$ relative RMS density error, the respective $(\epsilon_J,\epsilon_H)$ values are $(0.120,0.0175)$, $(0.243,0.0555)$, and $(1.219,0.825)$. The growth reflects error amplification in temporal differentiation, even when the physical transport is of potential form.

We next separate error sources using 60 native 1~s intervals on the same 100-BS geometry, 12 independent observation seeds, $h=20$~m, and nominal 2\% spatially correlated density noise. The staged estimates use, in turn, the physical kinetic reference, physical-mean Gaussian closure, canonical-mean Gaussian closure, and snapshot-based reconstruction. The final estimate includes a closure amplitude based on observed density and a variance estimated from 20\% of an independent auxiliary pool. Table~\ref{tab:error_decomposition} gives normalized adjacent-stage $L^1$ discrepancies and total error, using the same kinetic-reference denominator within each scenario. Unlike the boundary-integral bounds in \eqref{eq:directional_error_terms}, these are differences between successive estimates.
\begin{table}[t]
	\centering
	\caption{Controlled directional-error diagnosis under native 1~s evaluation. Stage discrepancies are not additive contributions to the total error.}
	\label{tab:error_decomposition} % Tab.III
	\footnotesize
	\vspace*{-2mm}
	\setlength{\tabcolsep}{4pt}
	\begin{tabular}{lcccc}
		\hline
		Mobility & Shape & Ident. & Recon. & Total\\
		\hline
		Potential coherent & 0 & 0 & 0.305 & 0.305\\
		Potential Gaussian & 0 & 0 & 0.222 & 0.225\\
		Nonpotential ($\alpha=0.5$) & 0 & 0.633 & 0.195 & 0.688\\
		Bimodal ($\rho=0.5$) & 0.202 & 0 & 0.510 & 0.551\\
		\hline
	\end{tabular}
	\vspace*{-5mm}
\end{table}

The nonpotential case shows a canonical-to-physical discrepancy of $0.633$, whereas balanced bimodal motion gives a Gaussian shape mismatch of $0.202$. Under noisy potential transport, the reconstruction discrepancies are $0.305$ for coherent motion and $0.222$ for Gaussian motion. The total also includes density-amplitude and variance-estimation effects. 
Unlike Fig.~\ref{fig:snapshot_interval}(c), which evaluates errors after temporal aggregation, Table~\ref{tab:error_decomposition} accumulates absolute errors over native $1$~s intervals. For the same $2\%$ noise realizations used in Fig.~\ref{fig:snapshot_interval}(c), these two evaluation orders yield $0.0555$ and $0.3021$, respectively.

\vspace*{-1mm}
\subsection{External Validation and Baseline Comparison}\label{subsec:external_validation} % S5.5

We evaluate F-STSG using the public FUMD-AI Krak\'ow vehicular dataset~\cite{Bernad2026Krakow}, generated with Simulation of Urban MObility (SUMO) and OMNeT++/Simu5G. The dataset contains vehicle trajectories, serving-cell records, and nine fixed BS locations. We construct nearest-BS Voronoi regions and use nonoverlapping $60$~s windows with fixed vehicle cohorts, obtaining 99 valid windows and 7523 geometric crossings. The reference crossings are extracted from consecutive $1$~s trajectory segments. Density snapshots are estimated from vehicle positions without using serving-cell records or crossing labels.

Reconstruction uses grid spacing $h\simeq d_{\mathrm{BS}}/30$ and Gaussian density bandwidth $d_{\mathrm{BS}}/10$, where $d_{\mathrm{BS}}$ is the median nearest-BS distance. 
The iterative solver converges on 86 of the 99 valid windows, giving net-flow and directional errors of $0.592$ and $0.606$, respectively, with an estimated-to-reference crossing-count ratio of $49.9\%$. An equivalent direct solver resolves all 99 windows without materially changing the original results. The baseline comparison uses 47 held-out windows with complete outputs from F-STSG and the three baselines.

We compare F-STSG with three baselines. Graph-MinNorm reconstructs pairwise flows from cell-population balance, while Graph-Weighted additionally accounts for Voronoi boundary lengths. Continuity-OF (continuity-constrained optical flow) estimates a spatial transport field using the population continuity constraint and smoothness regularization. All methods are evaluated on the same held-out windows and directed cell pairs. Unlike the graph-based baselines, F-STSG and Continuity-OF use spatial density information.

\begin{figure}[t]
	\centering
	\includegraphics[width=0.9\columnwidth]{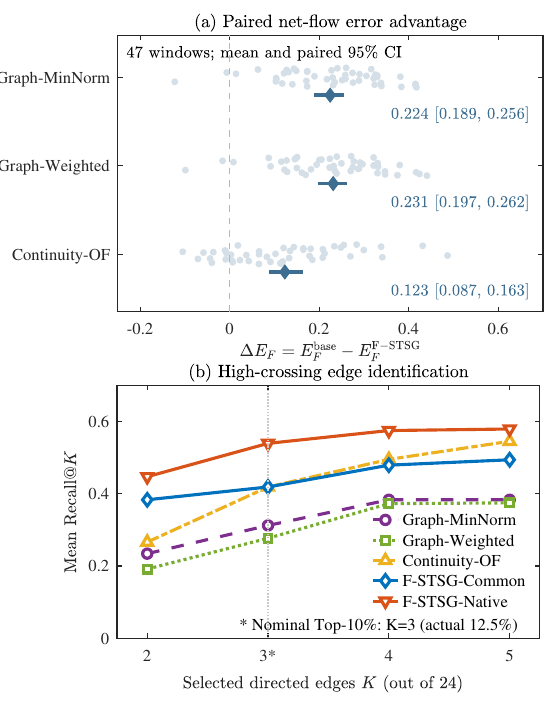}
	\vspace*{-5mm}
	\caption{Baseline comparison on 47 held-out windows. (a) Window-level paired net-flow error differences between each baseline and F-STSG. Dots indicate individual windows, while diamonds and horizontal intervals show mean differences and 95\% confidence intervals. Positive differences favor F-STSG. (b) Mean recall of high-crossing directed edges versus the selection budget $K$ among 24 eligible edges. The two F-STSG variants differ in their directional closure and aggregation order.}
	\label{fig:baseline_comparison} % Fig.7
	\vspace*{-5mm}
\end{figure}

Fig.~\ref{fig:baseline_comparison}(a) shows positive mean paired net-flow error differences of $0.224$, $0.231$, and $0.123$ against Graph-MinNorm, Graph-Weighted, and Continuity-OF, respectively, and all paired 95\% confidence intervals lie above zero. The pooled normalized net-flow error on the same held-out subset is nevertheless $0.585$, distinct from the window-level paired differences and indicating substantial remaining absolute error.

Fig.~\ref{fig:baseline_comparison}(b) compares two aggregation orders using the same reconstructed flux: F-STSG-Common takes the positive part after pairwise aggregation, whereas F-STSG-Native integrates the local positive part first. At $K=3$, Native improves high-crossing-edge recall from $0.418$ to $0.539$ and reduces directional error from $0.772$ to $0.611$. Thus, aggregation order affects directional estimates even when the reconstructed net flow is unchanged.

Across all 99 valid windows, raw Simu5G serving-cell transitions occur $2.567$ times as often as geometric crossings, whereas dataset-defined actionable events have a count ratio of $0.928$. Geometric crossings thus characterize mobility-related boundary activity rather than executed protocol handovers.

\begin{table}[t]
	\centering
\caption{Effects of temporal aggregation and kinetic closure on directional crossing errors.}
	\label{tab:temporal_kinetic} % Tab.IV
	\footnotesize
	\vspace*{-3mm}
	\setlength{\tabcolsep}{4pt}
	\begin{tabular}{cc}
		\hline
		\multicolumn{2}{c}{(a) Temporal aggregation}\\
		\hline
		Window (s) & Oracle $E_H$ (95\% CI)\\
		\hline
		1  & 0.068 [0.058, 0.079]\\
		5  & 0.235 [0.213, 0.256]\\
		15 & 0.415 [0.386, 0.441]\\
		58 & 0.584 [0.556, 0.608]\\
		\hline
	\end{tabular}
	
	\vspace{2mm}
	
	\setlength{\tabcolsep}{3pt}
	\begin{tabular}{lcccc}
		\hline
		\multicolumn{5}{c}{(b) Controlled kinetic closure: normalized error $E_H$}\\
		\hline
		Mobility & Kinetic & P-G & F-C & F-G\\
		\hline
		Low        & 0.0153 & 0.0153 & 0.0189 & 0.0155\\
		Gaussian   & 0.0246 & 0.0246 & 0.5332 & 0.0250\\
		Bimodal    & 0.0232 & 0.0559 & 0.5462 & 0.0574\\
		Correlated & 0.0256 & 0.0256 & 0.5426 & 0.0765\\
		\hline
	\end{tabular}
	
	\vspace{1mm}
	\parbox{\columnwidth}{\scriptsize
		P-G: Gaussian closure with physical mean transport;
		F-C: F-STSG with coherent closure;
		F-G: F-STSG with Gaussian closure and correct auxiliary variance.
		Kinetic denotes the physical kinetic reference evaluated against
		finite-sample trajectory counts.
	}
	\vspace*{-5mm}
\end{table}

\vspace*{-1mm}
\subsection{Temporal Aggregation and Kinetic Closure}\label{subsec:temporal_kinetic} % S5.6

Table~\ref{tab:temporal_kinetic}(a) examines temporal aggregation using trajectory-derived net crossings as oracle inputs on 47 held-out windows. Evaluation uses the common $58$~s interior interval of each $60$~s window, with 95\% confidence intervals. The coherent-closure error increases from $0.068$ at $1$~s to $0.584$ at $58$~s as opposite directional crossings cancel over longer intervals.

Table~\ref{tab:temporal_kinetic}(b) compares kinetic closures under controlled nonstationary mobility. P-G uses the physical mean transport, whereas F-G uses the prescribed true velocity variance, and these auxiliary quantities serve as oracle information for evaluating the closures. Under matched Gaussian mobility, F-G reduces the directional error from $0.5332$ (F-C) to $0.0250$, close to the physical kinetic reference. The errors increase to $0.0574$ under bimodal mobility and $0.0765$ under correlated mobility, reflecting distributional mismatch and, in the latter case, canonical-to-physical transport differences.

\vspace*{-1mm}
\section{Conclusion}\label{sec:conclusion} % S6

This paper has developed a kinetic formulation of geometric serving-region crossings in nonstationary cellular networks. By linking directional crossing rates to boundary-normal velocity moments, the formulation distinguishes signed net transfer from gross bidirectional crossings. We have established two information limits concerning physical transport recovery from density evolution and gross crossing recovery from first-order transport. Potential transport permits exact canonical flux recovery, while coherent and Gaussian closures characterize directional rates under specified mobility conditions. Sharp moment-constrained bounds further quantify crossing-rate uncertainty without prescribing a velocity distribution. The stationary Poisson--Voronoi crossing law follows as a special case.

Controlled experiments have examined transport ambiguity, snapshot reconstruction error, and velocity-closure mismatch. On held-out vehicular simulations, F-STSG achieves lower net-flow error than the considered baselines, although substantial absolute error remains. Auxiliary velocity statistics improve directional estimates under matched Gaussian mobility. Future work will extend the framework to radio-dependent association and handover mechanisms.

\bibliographystyle{IEEEtran}
\bibliography{ref}

\end{document}